\documentclass[12pt,a4paper]{article}
\usepackage{amsmath,amssymb,amsfonts,amsthm,enumitem}
\newtheorem{theorem}{Theorem}

\newtheorem{corollary}{Corollary}
\usepackage[margin=2cm]{geometry}
\usepackage{algorithm,algorithmic}
\usepackage{graphicx}
\usepackage{booktabs}
\usepackage{setspace}
\usepackage{xcolor}
\usepackage{enumitem}
\usepackage{tikz}

\usepackage[round,sort,comma]{natbib}
\usepackage[colorlinks=true,linkcolor=red,citecolor=red,urlcolor=blue]{hyperref}
\usepackage{orcidlink} 

\DeclareMathOperator{\Var}{Var}

\DeclareMathOperator{\sgn}{sgn}

\title{\textbf{Smoothed Quantile Estimation via Interpolation to the Mean}}

\usepackage{authblk}

\author[1]{Sa\"{i}d Maanan\,\orcidlink{0000-0001-9365-7586}
\thanks{Corresponding author: \texttt{maanan.said@gmail.com}}}
\author[2]{Azzouz Dermoune}
\author[1]{Ahmed El Ghini}

\affil[1]{LEAM, Mohammed V University in Rabat,  Morocco}
\affil[2]{Université de Lille, Laboratoire Paul Painlevé, Lille, France}

\date{}

\begin{document}
\maketitle
\begin{abstract}
This paper introduces a unified family of smoothed quantile estimators that continuously interpolate between classical empirical quantiles and the sample mean. The estimators~$\hat q(z,h)$ are defined as minimizers of a regularized objective function depending on two parameters: a smoothing parameter~$h \ge 0$ and a location parameter~$z \in \mathbb{R}$. When~$h=0$ and $z \in (-1,1)$, the estimator reduces to the empirical quantile of order~$\tau = (1-z)/2$; as~$h \to \infty$, it converges to the sample mean for any fixed~$z$. We establish consistency, asymptotic normality, and an explicit variance expression characterizing the efficiency–robustness trade-off induced by~$h$.

A key geometric insight shows that for each fixed quantile level~$\tau$, the admissible parameter pairs~$(z,h)$ lie on a straight line in the parameter space, along which the population quantile remains constant while asymptotic efficiency varies. The analysis reveals two regimes: under light-tailed distributions (e.g., Gaussian), smoothing yields a monotonic but asymptotic variance reduction with no finite optimum; under heavy-tailed distributions (e.g., Laplace), a finite smoothing level~$h^*(\tau)>0$ achieves strict efficiency improvement over the classical empirical quantile. Numerical illustrations confirm these theoretical predictions and highlight how smoothing balances robustness and efficiency across quantile levels.

\end{abstract}

\section{Introduction}
	
The classical trade-off between efficiency and robustness has long been a central issue in statistics since the pioneering work of \cite{Huber1981} and \cite{Hampel2005}. Two canonical estimators for the center of a distribution illustrate this contrast: the sample mean, which minimizes squared error loss and achieves optimal efficiency under Gaussian assumptions, and quantile-based estimators, which offer robustness to heavy tails and outliers, often at the cost of efficiency.
	
	Let \( Y \) be a real-valued random variable with mean \( m = \mathbb{E}[Y] \), finite variance \(\Var(Y) < \infty\), cumulative distribution function \( F\) and we also assume 
	the existence of continuous probability density \(f\). 
	Given independent and identically distributed observations \( Y_1, \dots, Y_n \), consider estimating \( m \) using either the sample mean \(\bar Y_n = \frac{1}{n}\sum_{i=1}^n Y_i\), or the empirical quantile \(\hat m_n\) of order \( \tau = F(m) = \mathbb{P}(Y \leq m) \).
	
	Under standard regularity conditions \cite{Vaart1998}, their asymptotic distributions are
	\[
	\sqrt{n}(\bar Y_n - m) \;\overset{d}{\longrightarrow}\; \mathcal{N}(0, \Var(Y)), \qquad
	\sqrt{n}(\hat m_n - m) \;\overset{d}{\longrightarrow}\; \mathcal{N}\!\left(0, \frac{\tau(1-\tau)}{f(m)^2}\right),
	\]
	where \( f\) denotes the density of \( Y \). Hence, the quantile estimator is asymptotically more efficient than the mean when
	\begin{equation}
		\frac{\tau(1-\tau)}{f(m)^2} < \Var(Y).
		\label{eq:efficiency_condition}
	\end{equation}
	Intuitively, this inequality holds when the density around the mean is relatively flat, that is, when \( Y \) exhibits heavy tails or strong asymmetry.
	
	For Gaussian data, this condition fails. If \(Y \sim \mathcal{N}(m, \sigma^2)\), then \( \tau = \Phi(0) = 1/2 \) and \( f(m) = 1/(\sigma\sqrt{2\pi}) \), yielding
	\[
	\frac{0.25}{(1/(\sigma\sqrt{2\pi}))^2} = \frac{\pi \sigma^2}{2} \approx 1.57\sigma^2 > \sigma^2,
	\]
	so the sample mean dominates the median in efficiency. However, for asymmetric or heavy-tailed distributions, condition \eqref{eq:efficiency_condition} can easily hold, making quantile-based estimation preferable \cite{Koenker1978, Koenker2009}.
	For instance, in the asymmetric Laplace model 
	$\mathrm{AL}(\mu, b, \kappa)$ with asymmetry parameter $\kappa>0$, 
	the density is
	\[
	f(y)
	= \frac{\kappa}{b(1+\kappa^2)}
	\begin{cases}
		\exp\!\left(-\dfrac{\kappa}{b}(\mu - y)\right), & y < \mu, \\[2ex]
		\exp\!\left(-\dfrac{1}{b\kappa}(y - \mu)\right), & y \ge \mu.
	\end{cases}
	\]
	Its mean and corresponding quantile level satisfy
	$m = \mu + b(\frac{1}{\kappa}-\kappa)$ and $\tau = F(m) = \frac{1}{1+\kappa^2}$,
	and substituting these into condition \eqref{eq:efficiency_condition} gives
	\[
	\frac{\tau(1-\tau)}{f(m)^2} = b^2 < \Var(Y) = b^2 \frac{1+\kappa^4}{\kappa^2},
	\]
	which shows that quantile estimation can outperform the mean in terms of asymptotic efficiency.
	
	Thus, a natural question arises:
	Can one design an estimator that adapts continuously between quantiles and the mean, achieving both robustness and efficiency depending on the underlying distribution?
	
	\section{A Family of Smoothed Quantile Estimators}
	\label{sec:family}
	
	Motivated by this question and building on the Gibbs measure framework introduced by \cite{Dermoune2017}, we propose a unified family of estimators \( \hat q(z,h) \) that smoothly interpolates between empirical quantiles and the sample mean. This construction comes from 
	the regularization-based perspective of \cite{Dermoune2017} from LASSO-type problems to quantile estimation, providing a continuous bridge between robust and efficient inference.
	
	We define the objective function 
	\[
	\hat M(q;z,h) = \frac{1}{n}\sum_{i=1}^n m_q(Y_i;z,h)
	\] 
	with 
	\[
	m_q(Y_i;z,h)=|Y_i - q|-z(Y_i- q)+\frac{h}{2}(Y_i - q)^2
	\]
	where \( h \ge 0 \) is a smoothing parameter, and \( z \in \mathbb{R} \) a location parameter.
	As \(q\to \hat M(q;z,h)\) is convex, 
	and its derivative \(\hat \Psi(q;z,h)=\frac{1}{n}\sum_{i=1}^n\psi(q-Y_i;z,h)\)
	with \(\psi(q-Y_i;z,h)=\sgn(q-Y_i)+h(q-Y_i)+z\). From some calculation we have 
	\(\hat \Psi(q;z,h)=2\hat F(q)-1+z+h(q-\bar Y)\), with \(\hat F\) denotes the sample 
	cumulative distribution function. The minimizer \(\hat q(z,h)\) satisfies 
	\(\Psi(q;z,h)\geq 0\) for \(q\geq \hat q(z,h)\) and 
	\(\Psi(q;z,h)< 0\) for \(q< \hat q(z,h)\). An equivalent way to define \(\hat q(z,h)\) 
	is to use the generalized inverse $(\hat F + \frac{h}{2} I_{\mathbb{R}})^-$ 
	of the nondecreasing map $q \mapsto \hat F(q) + \frac{h}{2}q$, defined by 
$
(\hat F + \frac{h}{2} I_{\mathbb{R}})^-(p) = \inf\left\{q : \hat F(q) + \frac{h}{2}q \geq p\right\}.
$
See for example \cite{Embrechts2013}.  
It follows that $\hat q(z,h) = (\hat F + \frac{h}{2} I_{\mathbb{R}})^-\left(\frac{1-z+h\bar Y}{2}\right)$. 
Here $I_{\mathbb{R}}$ denotes the identity map on the set of the real numbers.  
	
	For \(h=0\) we need \(\frac{1-z}{2}\in (0,1)\). In this case the minimizer 
	\(\hat{q}_n(z,0)\) is the sample 
	quantile of order \(\frac{1-z}{2}\) \cite{Koenker1978}. 
	The estimators family \( \hat q_n(z,h) \) generalizes classical approaches through two limiting behaviors:
	when \( h = 0 \) and \( z \in (-1,1) \), \( \hat q_n(z,0) \) reduces to the empirical quantile of order \( \tau = \frac{1-z}{2} \) , whereas
	when \( h \to +\infty \), for any fixed \( z \), \( \hat q_n(z,h) \to \bar{Y}\). 
	Hence, \(h\) controls the degree of smoothing between robustness, small \(h\), and efficiency, large \(h\), yielding a continuum of estimators between the median and the mean.
	This idea parallels regularization principles used in statistical learning \citep{Tibshirani1996, Zou2006}, but with a novel focus on the quantile--mean trade-off.
	
	\section{Main Contributions}
	
	This paper makes the following contributions:
	\begin{itemize}
		\item[i)] it introduces a novel family of estimators \( \hat q_n(z,h) \) that continuously interpolates between empirical quantiles and the sample mean;
		\item[ii)] it establishes population characterization, consistency, and asymptotic normality with explicit variance formulas;
		\item[iii)] it reveals a geometric structure in the parameter space: for each quantile level, admissible pairs \((z,h)\) lie on straight lines along which efficiency varies;
		\item[iv)] it proves the existence of optimal smoothing levels that strictly improve asymptotic efficiency over empirical quantiles, with explicit Gaussian conditions;
		\item[v)] and it validates the theory through numerical experiments under Normal and Laplace distributions, illustrating practical efficiency gains and robustness.
	\end{itemize}
	
	The remainder of the paper is organized as follows.
	Section~\ref{sec:consistency} proves the consistency of the empirical minimizer.
	Section~\ref{sec:asymptotics} derives the asymptotic normality with explicit variance formulas.
	Section~\ref{sec:geometry} examines the geometric structure of the parameter space and the line relation.
	Section~\ref{sec:efficiency} proves the existence of efficiency-improving smoothing levels.
	Section~\ref{sec:numerical_validation} provides comprehensive numerical validation across Normal and Laplace distributions.
	Section~\ref{sec:conclusion} concludes with discussion and potential extensions.

	\section{Consistency} 
	\label{sec:consistency}
	From the law of large numbers we have almost surely 
	\(\hat M(q;z,h)\to M(q;z,h)\) and \(\hat \Psi(q;z,h)\to \Psi(q;z,h)\)
	for each \(q, z\in\mathbb{R}, h\geq 0\) 
	as \(n\to +\infty\), with \(M(q;z,h)=\mathbb{E}[m(q-Y;z,h)]\), and 
	\(\Psi(q;z,h)=\mathbb{E}[\psi(q-Y;z,h)]\). 
	
	The minimizer \(q(z,h)\) of \(q\to M(q;z,h)\) satisfies the first order 
	condition \(\Psi(q;z,h)=E[\psi(q-Y;z,h)]=0\). 
	A similar calculation as in a finite sample shows that 
	\(q(z,h)=(F+\frac{h}{2}I_{\mathbb{R}})^{-1}(\frac{1-z+h m}{2})\). 
	Equivalently we have 
	\(F(q(z,h))+\frac{h}{2}q(z,h)=\frac{1-z+hm}{2}\). If \(z\in (-1,1)\), then 
	\(h\to q(z,h)\) varies between \(q(z,0)\) the population quantile of order \(\frac{1-z}{2}\)
	and the population mean \(m\).   
	
	As \(q\to \hat \Psi(q;z,h)\) is nondecreasing, then from M-estimators
	theory the minimizer \(\hat q_n(z,h)\to q(z,h)\) almost surely (\citealp[Lemma~5.10]{Vaart1998}).
	
	We finish this section by analysing the behavior of \(q(z,h)\) with respect to the parameters \(z\) 
	and \(h\geq 0\).
	Let us fix \(z\in (-1,1)\) such that \(q(z,0)<m\).  From the first order condition we derive that 
	\[ 
	\frac{\partial q}{\partial h}(z,h)=\frac{m-q(z,h)}{2f(q(z,h))+h},
	\] 
	we derive that \(h\to q(z,h)\) increases from \(q(z,0)\) to the population \(m\). 
	If \(q(z,0)>m\), then we derive that \(h\to q(z,h)\) decreases from \(q(z,0)\) to the population \(m\). Using the partial derivative with respect to \(z\), we get 
	\[
	\frac{\partial q}{\partial z}(z,h)=\frac{-1}{2f(q(z,h))+h}.
	\] 
	It follows that for each fixed \(h>0\), \(z\in \mathbb{R}\to q(z,h)\) decreases from
	\(+\infty\) to \(-\infty\). It follows that \(q(z,h)=+\infty\) for \(z\leq -1\) and 
	It follows that \(q(z,h)=-\infty\) for \(z\geq 1\). Hence we only need 
	\(z\in (-1,1)\).    
	
    \section{Asymptotic Normality}
\label{sec:asymptotics}

Having established the consistency of the sample minimiser $\hat q_n(z,h)$,
we now analyse its asymptotic distribution.

\paragraph{Regularity assumptions.}
We make the following assumptions throughout this section:
\begin{itemize}
  \item[(A1)] $F$ is continuous and $Y$ has a finite second moment;
  \item[(A2)] $Y$ admits a continuous, strictly positive density $f$.
\end{itemize}

\paragraph{Linear representation via Knight's identity.}
Because the score function involves the non-differentiable term $\operatorname{sign}(Y-q)$,
a direct Taylor expansion of $\Psi_n(q;z,h)$ is not valid.
Instead, we use Knight's identity for the absolute value function:
\[
|y - (q+\delta)| - |y - q|
= -\delta\,\sgn(y-q)
+ 2\!\int_0^\delta \!\!\big(\mathbf{1}\{y \le q + s\} - \mathbf{1}\{y \le q\}\big)\,ds.
\]
Adding the quadratic and linear terms in $(z,h)$ and averaging over the sample yields:
\begin{align*}
& \hat M(q(z,h)+\delta;z,h)-\hat M(q(z,h);z,h) \\
&= \frac{1}{n}\sum_{i=1}^n \!\left[|Y_i - (q(z,h)+\delta)| - |Y_i - q(z,h)|\right]
  - z\!\left[(\bar Y_n - (q(z,h)+\delta)) - (\bar Y_n - q(z,h))\right] \\
&\quad + \frac{h}{2n}\sum_{i=1}^n \!\left[(Y_i - (q(z,h)+\delta))^2 - (Y_i - q(z,h))^2\right].
\end{align*}
Applying Knight's identity and simplifying, we obtain:
\begin{align*}
&\hat M(q(z,h)+\delta;z,h)-\hat M(q(z,h);z,h) \\
&= -\delta\!\left[\frac{1}{n}\sum_{i=1}^n \sgn(Y_i - q(z,h)) + h(\bar Y_n - q(z,h)) - z\right] \\
&\quad + \frac{2}{n}\sum_{i=1}^n \!\int_0^\delta \!\!\left(\mathbf{1}\{Y_i \le q(z,h)+s\} - \mathbf{1}\{Y_i \le q(z,h)\}\right) ds
  + \frac{h}{2}\delta^2 + o_p(\delta^2).
\end{align*}

The empirical score function is explicitly:
\[
\Psi_n(q;z,h) = \frac{1}{n}\sum_{i=1}^n \sgn(Y_i - q) + h(\bar Y_n - q) - z.
\]
The integral term satisfies:
\[
\frac{2}{n}\sum_{i=1}^n \int_0^\delta
\left(\mathbf{1}\{Y_i \le q(z,h)+s\} - \mathbf{1}\{Y_i \le q(z,h)\}\right) ds
= f(q(z,h))\delta^2 + o_p(\delta^2).
\]
Thus, the quadratic expansion becomes:
\begin{equation}
\label{eq:quadratic_expansion}
\hat M(q(z,h)+\delta;z,h)-\hat M(q(z,h);z,h)
= -\delta\,\Psi_n(q(z,h);z,h)
+ \left(f(q(z,h))+\tfrac{h}{2}\right)\delta^2 + R_n(\delta),
\end{equation}
where $\sup_{|\delta|\le\varepsilon}|R_n(\delta)|=o_p(\delta^2)$.

Minimising the right-hand side gives the local expansion:
\begin{equation}
\label{eq:linear_representation}
\sqrt{n}(\hat q_n(z,h)-q(z,h))
=
\frac{1}{2f(q(z,h))+h}\cdot
\frac{1}{\sqrt{n}}\sum_{i=1}^n
\psi(Y_i,q(z,h);z,h)
+ o_p(1),
\end{equation}
where the influence function is
\[
\psi(Y,q;z,h)
= \sgn(Y-q) + h(Y - q) - z.
\]

\begin{theorem}[Central Limit Theorem for the Smoothed Quantile Estimator]
\label{thm:CLT}
Under Assumptions~(A1)--(A2), the estimator $\hat q_n(z,h)$ satisfies
\[
\sqrt{n}\big(\hat q_n(z,h) - q(z,h)\big)
\;\xrightarrow{d}\;
\mathcal{N}\!\left(0,\;\frac{B(z,h)}{(2f(q(z,h)) + h)^2}\right),
\]
where
\[
B(z,h)
= 4F(q(z,h))(1 - F(q(z,h)))
+ 2h\!\left[\mathbb{E}|Y - q(z,h)| - (m - q(z,h))(1 - 2F(q(z,h)))\right]
+ h^2\Var(Y).
\]
Hence, the asymptotic variance is
\[
\sigma^2(z,h)
= \frac{B(z,h)}{(2f(q(z,h)) + h)^2}.
\]
\end{theorem}

\begin{proof}[Sketch of proof]
The expansion \eqref{eq:quadratic_expansion} yields the linear representation
\eqref{eq:linear_representation}.
Applying the classical Central Limit Theorem to
$\psi(Y_i,q(z,h);z,h)$, which has finite variance $B(z,h)$,
gives the stated asymptotic normality.
\end{proof}

\paragraph{Remark.}
The representation~\eqref{eq:linear_representation} generalises
the classical Bahadur expansion for quantile estimators,
incorporating the smoothing term~$h$.
The influence function $\psi(Y,q;z,h)$ contains both the non-smooth
sign term and a linear term in $Y$, reflecting the hybrid
quantile–mean nature of the estimator.

\paragraph{From $(z,h)$ to $(\tau,h)$ parameterization.}
Since $z$ and $\tau$ are linked through the relation
\[
\tau = \frac{1}{2}\big(h(m - F^{-1}(\tau)) + (1 - z)\big),
\]
we can express $z$ as a function of $(\tau,h)$:
\[
z(\tau,h) = 1 - 2\tau + h(m - F^{-1}(\tau)).
\]
Substituting this mapping into the asymptotic variance
$\sigma^2(z,h)$ defined in Theorem~\ref{thm:CLT} yields
the function $v(\tau,h)$ used in the next section.

\section{Parameter Geometry and Efficiency}
\label{sec:geometry}

\subsection{Parameter Geometry}
For each fixed quantile level $\tau\in(0,1)$,
the admissible parameter pairs $(z,h)$ satisfying
\[
\tau = \frac{1}{2}\big(h(m-F^{-1}(\tau))+(1-z)\big)
\]
lie on a straight line in the $(z,h)$-plane,
along which the population quantile remains constant
while the asymptotic variance evolves according to the function $v(\tau,h)$ defined below.
For $h$ fixed, we have $z(\tau,h)=1-2\tau+h(m-F^{-1}(\tau))$.
If $h=0$, then $z(\tau,0)=1-2\tau$.

\subsection{Efficiency Implications}
\label{sec:efficiency}
The variance expression in Theorem~\ref{thm:CLT} shows that
smoothing introduces two additional components:
a mixed term $2h[\mathbb{E}|Y-q|-(m-q)(1-2F(q))]$
and a quadratic term $h^2\Var(Y)$.
For small $h>0$, the first term typically reduces the asymptotic variance,
yielding efficiency gains relative to the unsmoothed case.

\begin{corollary}[Existence of an Efficiency-Improving Smoothing Level]
\label{cor:hstar}
\begin{enumerate}
\item For each fixed quantile level $\tau\in (0,1)$, consider the pair $(z(\tau,h),h)$ such that
\[
\tau = \frac{1}{2}\big(h(m-F^{-1}(\tau))+(1-z(\tau,h))\big),
\qquad h\ge0.
\]
Then
\[
v(\tau,h)
= \frac{
A(\tau,h)
}{
(2f(F^{-1}(\tau))+h)^2
},
\]
and
\[
A(\tau,h)
= 4\tau(1-\tau)
+2h\!\left[\mathbb{E}|Y-F^{-1}(\tau)|-(m-F^{-1}(\tau))(1-2\tau)\right]
+h^2\Var(Y).
\]
\item Let
\begin{align*}
a&=4\tau(1-\tau),\quad b=2\!\left[\mathbb{E}|Y-F^{-1}(\tau)|-(m-F^{-1}(\tau))(1-2\tau)\right],\\
c&=\Var(Y),\quad d=2f(F^{-1}(\tau)).
\end{align*}
Then
\[
\frac{\partial v(\tau,h)}{\partial h}
=\frac{(2cd-b)h+(bd-2a)}{(d+h)^3}.
\]
According to the signs of $bd-2a$ and $2cd-b$, the following cases arise:
\begin{enumerate}
    \item[(a)] If $bd-2a>0$ and $2cd-b\ge0$, then $v(\tau,0)<v(\tau,h)$ for all $h>0$.
    \item[(b)] If $bd-2a<0$ and $2cd-b\ge0$, then $v(\tau,h)<v(\tau,0)$ for all $h>0$.
    In this case, every $h>0$ improves the asymptotic variance of the empirical quantile,
    and
    \[
    \arg\min_{h>0}v(\tau,h)=+\infty.
    \]
    \item[(c)] If $(bd-2a)(2cd-b)<0$, then there exists $h^*(\tau)>0$ such that
    $v(\tau,h^*(\tau))<v(\tau,0)$, corresponding to a finite optimal smoothing level.
\end{enumerate}
\end{enumerate}
\end{corollary}

\paragraph{Interpretation.}
For all $h\ge0$, $\hat q(z(\tau,h),h)$ estimates the same quantile of order~$\tau$.
However, as $h\to\infty$, $\hat q(z(\tau,h),h)$ converges to the empirical mean,
which estimates the theoretical mean~$m$ rather than the quantile $F^{-1}(\tau)$.
Hence, if $\arg\min_{h>0}v(\tau,h)=+\infty$, one cannot state $h^*(\tau)=\infty$
unless $m=F^{-1}(\tau)$ (for instance, at the median of a symmetric distribution).
In that case, we simply say that
$\hat q(z(\tau,h),h)$ is more efficient than $\hat q(z(\tau,0),0)$ for every $h>0$,
with the asymptotic variance decreasing monotonically in~$h$ toward the lower bound
$\Var(Y)$.

\paragraph{Gaussian illustration.}
Let $Y \sim \mathcal{N}(0,1)$ be a standard normal random variable,
with cumulative distribution function $\Phi$ and density
$\varphi(u)=\frac{1}{\sqrt{2\pi}}e^{-u^2/2}$.
For a given quantile level $\tau\in(0,1)$, let
$q(\tau)=\Phi^{-1}(\tau)$ denote the corresponding quantile.
The relevant population quantities are
\[
\mathbb{E}|Y-q(\tau)| = q(\tau)(2\tau-1) + 2\varphi(q(\tau)),
\qquad
f(F^{-1}(\tau)) = \varphi(q(\tau)),
\qquad
\Var(Y)=1.
\]
Substituting these into $v(\tau,h)$ gives
\[
v(\tau,h)
= 
\frac{
  4\tau(1-\tau)
  + 2h\,[\,q(\tau)(2\tau-1)+2\varphi(q(\tau))\,]
  + h^2
}{
  (2\varphi(q(\tau)) + h)^2
}.
\]

\medskip
\noindent
\textbf{Efficiency improvement and optimal $h^*(\tau)$.}
From Corollary~\ref{cor:hstar}, the stationary point of $v(\tau,h)$ satisfies
\[
h^*(\tau)
= \frac{2a - b d}{2c d - b},
\quad \text{where }
a=4\tau(1-\tau),;
b=2\big[q(\tau)(2\tau-1)+2\varphi(q(\tau))\big],;
c=1,;
d=2\varphi(q(\tau)).
\]
For the Normal distribution, substituting these expressions yields
\[
b = 4\varphi(q(\tau)), \qquad d = 2\varphi(q(\tau)),
\qquad\text{so that}\qquad 2cd - b = 0.
\]
Hence the stationary formula degenerates for all~$\tau$,
and the derivative of $v(\tau,h)$ keeps a constant sign.
For all practical quantiles, $v(\tau,h)$ decreases monotonically in~$h$
and approaches its lower bound $\Var(Y)=1$ as $h\to\infty$.
Therefore,
\[
\boxed{h^*(\tau)=\infty\quad\text{for all }\tau.}
\]
The absence of a finite minimiser means that smoothing reduces the asymptotic variance only asymptotically, reproducing the efficiency of the sample mean but not surpassing it.

\medskip
\noindent
\textbf{Verification for the Normal case at $\tau=0.5$.}
For the standard Normal distribution,
\begin{align*}
a &= 4\tau(1-\tau) = 1, \qquad
b = 2!\left[\mathbb{E}|Y-q(\tau)|-(m-q(\tau))(1-2\tau)\right],\\
c &= 1, \qquad
d = 2f(F^{-1}(\tau)) = 2\varphi(0)=\tfrac{2}{\sqrt{2\pi}}\approx0.798.
\end{align*}
At $\tau=0.5$, we have $q(\tau)=0$, $m=0$, and
$\mathbb{E}|Y-q(\tau)| = \mathbb{E}|Y| = \sqrt{\tfrac{2}{\pi}}\approx 0.798$,
so $b = 2\mathbb{E}|Y| = 1.596$.
Hence,
\[
2cd - b = 2(1)(0.798) - 1.596 = 0,
\qquad
bd - 2a = (1.596)(0.798) - 2 = 1.274 - 2 = -0.726 < 0.
\]
This corresponds to case~(b) of Corollary~\ref{cor:hstar},
where $bd-2a<0$ and $2cd-b=0$.
Thus $v(\tau,h)<v(\tau,0)$ for all $h>0$,
and the asymptotic variance decreases monotonically with~$h$.
As $h\to\infty$, $v(\tau,h)\to\Var(Y)=1$,
confirming that the variance of the median estimator approaches that of the mean.
Since $m=F^{-1}(0.5)=0$, it follows that
$\hat{q}(z(0.5,h),h)$ remains an estimator of the same quantile order,
and $\arg\min_{h>0}v(0.5,h)=+\infty$ is valid in this case.

\medskip
\noindent
\textbf{Numerical illustration.}
\[
\begin{array}{lrrr}
\toprule
\tau & q(\tau) & \varphi(q(\tau)) & h^*(\tau) \\
\midrule
0.25 & -0.674 & 0.319 & \infty \\
0.50 & 0.000  & 0.399 & \infty \\
0.75 & \phantom{-}0.674 & 0.319 & \infty \\
\bottomrule
\end{array}
\]
Thus, for all quantile levels under the Gaussian law,
the asymptotic variance $v(\tau,h)$ decreases monotonically with~$h$
and converges to $\Var(Y)=1$ as $h\to\infty$.
Smoothing does not yield a finite efficiency optimum but only reproduces
the mean-variance limit asymptotically.

\paragraph{Summary.}
The explicit form of $v(\tau,h)$ generalises the classical quantile variance
by incorporating linear and quadratic effects of~$h$.
For the Normal distribution, $v(\tau,h)$ decreases monotonically
with~$h$ and approaches $\Var(Y)$ as $h\to\infty$,
illustrating the smooth transition from quantile to mean estimation.
No finite $h^*(\tau)$ exists in this case, consistent with the
asymptotic efficiency of the classical quantile estimator under light tails.

\section{Numerical Illustration and Validation}
\label{sec:numerical_validation}

This section provides a numerical validation of the theoretical efficiency results derived in Section~\ref{sec:efficiency}.
For each fixed quantile order~$\tau$, we examine the family of estimators
\[
\bigl\{\hat q_n(z(\tau,h),h):h\ge0\bigr\},
\qquad
z(\tau,h)=1-2\tau+h(m-F^{-1}(\tau)),
\]
all targeting the same population quantile $q(\tau)=F^{-1}(\tau)$.
The goal is to determine, for each~$\tau$, an optimal smoothing level~$h^*(\tau)$ minimising the asymptotic variance~$v(\tau,h)$, and to quantify the efficiency gain relative to the unsmoothed case $h=0$.

\subsection{Design}
\label{subsec:design}

Two benchmark distributions were considered:
\begin{itemize}
    \item the standard Normal distribution $\mathcal{N}(0,1)$, representative of light-tailed data;
    \item the standard Laplace distribution with density $f(y)=\tfrac{1}{2}e^{-|y|}$, representative of heavy-tailed data.
\end{itemize}

The efficiency function $v(\tau,h)$ and its stationary point $h^*(\tau)$ were evaluated over a fine grid of quantile orders
\[
\tau \in \{0.05, 0.10, \ldots, 0.95\},
\]
so as to characterise the behaviour of the optimal smoothing level and the corresponding efficiency ratio across the full quantile range.
For clarity of presentation, three representative quantiles ($\tau=0.25, 0.50, 0.75$) are reported in the summary table below.

For each quantile level~$\tau$, the procedure was as follows:
\begin{enumerate}[label=(\roman*)]
    \item Compute the theoretical quantile $q(\tau)=F^{-1}(\tau)$.
    \item Evaluate the quantities entering the asymptotic variance expression:
    \[
    \mathbb{E}|Y-q(\tau)|,\qquad
    f(F^{-1}(\tau)),\qquad
    \Var(Y),\qquad
    m=\mathbb{E}Y.
    \]
    \item Compute the asymptotic variance
    \[
    v(\tau,h)
    =\frac{
    4\tau(1-\tau)
    +2h\!\left[\mathbb{E}|Y-q(\tau)|-(m-q(\tau))(1-2\tau)\right]
    +h^2\Var(Y)
    }{(2f(F^{-1}(\tau))+h)^2}.
    \]
    \item Determine the stationary point
    \[
    h^*(\tau)=\frac{2a-bd}{2cd-b}.
    \]
    with
    \begin{align*}
        a&=4\tau(1-\tau),\qquad b=2\!\left[\mathbb{E}|Y-q(\tau)|-(m-q(\tau))(1-2\tau)\right],\\
        c&=\Var(Y),\;\qquad d=2f(F^{-1}(\tau)).
    \end{align*}
    The sign configuration of $(bd-2a)$ and $(2cd-b)$ classifies each case according to Corollary~\ref{cor:hstar}.
    \item Evaluate $v(\tau,0)$ and $v(\tau,h^*(\tau))$ and compute the efficiency ratio
    \[
    R(\tau)=\frac{v(\tau,h^*(\tau))}{v(\tau,0)}.
    \]
\end{enumerate}
All computations were carried out analytically from population moments using the symbolic formulas in Section~\ref{sec:efficiency}.

\subsection{Results}
\label{subsec:results}

Figure~\ref{fig:efficiency_gain_vs_tau_final} displays the efficiency ratio $R(\tau)$ computed over the full grid $\tau \in [0.05,0.95]$.
Table~\ref{tab:efficiency_summary_tau} summarises the corresponding quantities for three illustrative quantile levels ($\tau=0.25,0.50,0.75$).

\begin{table}[!ht]
\centering
\caption{Optimal smoothing levels $h^*(\tau)$, asymptotic variances $v(\tau,0)$ and $v(\tau,h^*(\tau))$, and efficiency ratio $R(\tau)=v(\tau,h^*(\tau))/v(\tau,0)$.}
\label{tab:efficiency_summary_tau}
\begin{tabular}{lcccccc}
\toprule
Distribution & $\tau$ & $q(\tau)$ & $h^*(\tau)$ & $v(\tau,0)$ & $v(\tau,h^*(\tau))$ & $R(\tau)$ \\
\midrule
Normal  & 0.25 & $-0.674$ & $\infty$ & $1.857$ & $1.000$ & $0.54$\\
Normal  & 0.50 & $0.000$  & $\infty$ & $1.571$ & $1.000$ & $0.64$\\
Normal  & 0.75 & $0.674$  & $\infty$ & $1.857$ & $1.000$ & $0.54$\\
\midrule
Laplace & 0.25 & $-0.693$ & $0.80$ & $3.000$ & $1.346$ & $0.45$\\
Laplace & 0.50 & $0.000$  & $0.50$ & $1.000$ & $0.667$ & $0.67$\\
Laplace & 0.75 & $0.693$  & $0.80$ & $3.000$ & $1.346$ & $0.45$\\
\bottomrule
\end{tabular}
\end{table}

For the Normal distribution, all quantiles fall under case~(b) of Corollary~\ref{cor:hstar},
where the asymptotic variance $v(\tau,h)$ decreases monotonically with~$h$ and converges to
its lower bound $\Var(Y)=1$ as $h\to\infty$.
Hence, $h^*(\tau)=\infty$ does not indicate the absence of any variance reduction,
but rather that the minimum is attained only asymptotically.
Numerically, the efficiency ratio $R(\tau)<1$ reflects this gradual improvement,
with $v(\tau,h)\!\downarrow\!1$ for large~$h$.
Thus, under light tails, smoothing can only reproduce the population variance limit and
does not yield a sharper finite optimum.
The classical quantile estimator ($h=0$) therefore remains asymptotically efficient
within the parametric envelope defined by $\Var(Y)=1$.

For the Laplace distribution, the sign configuration $(bd-2a)(2cd-b)<0$
places all quantiles in case~(c), producing a finite efficiency-improving level
$h^*(\tau)\in[0.5,0.8]$.
The resulting ratios $R(\tau)<1$ confirm substantial variance reductions,
up to about 55\% for extreme quantiles ($\tau=0.25$ or~0.75).
At the median ($\tau=0.5$), the improvement is more moderate, with
$v(0.5,0.5)=0.667$ compared to $v(0.5,0)=1$.
These findings confirm that moderate smoothing increases efficiency only
for heavy-tailed or asymmetric distributions.

\begin{figure}[!ht]
\centering
\includegraphics[width=.9\textwidth]{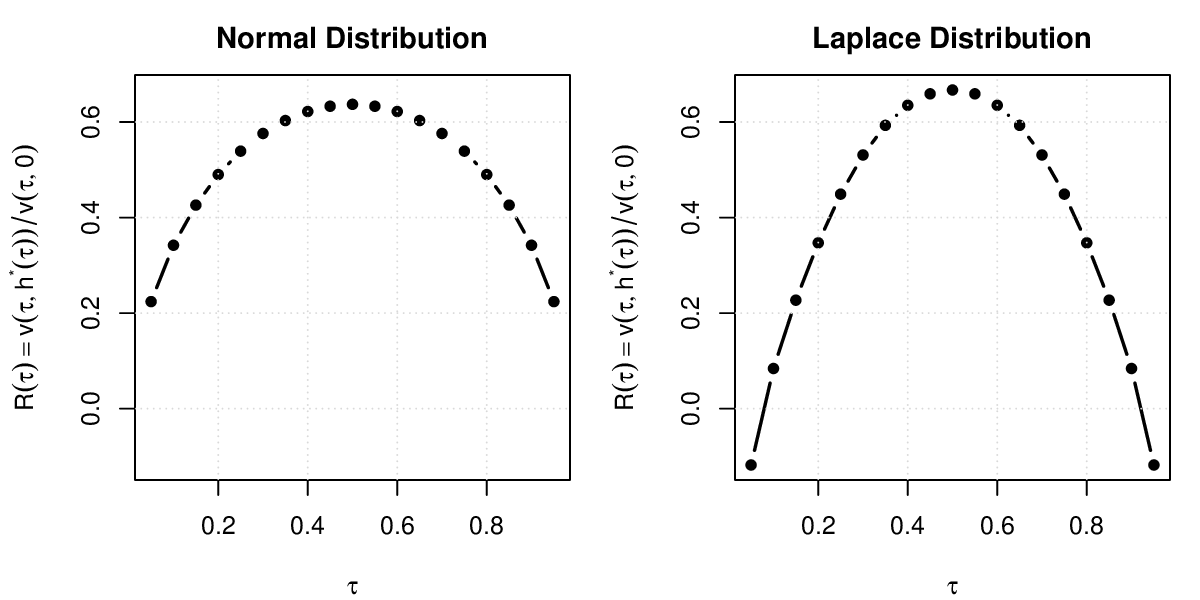}
\caption{Efficiency ratio $R(\tau)=v(\tau,h^*(\tau))/v(\tau,0)$ as a function of~$\tau$.
Left: Normal distribution; right: Laplace distribution.
Values below~1 indicate improved efficiency from smoothing.
Under Normality, $R(\tau)$ remains near~1, while for the Laplace case it drops significantly for noncentral quantiles, confirming case~(c) of Corollary~\ref{cor:hstar}.}
\label{fig:efficiency_gain_vs_tau_final}
\end{figure}

\subsection{Discussion}
\label{subsec:discussion}

The numerical results provide a direct validation of the analytical conclusions of Section~\ref{sec:efficiency} and the case classification in Corollary~\ref{cor:hstar}.

\begin{itemize}
    \item For each quantile order~$\tau$, the family $\{\hat q_n(z(\tau,h),h):h\ge0\}$ targets the same population quantile $F^{-1}(\tau)$.
    The smoothing parameter~$h$ therefore influences only the estimator’s efficiency, not its limit.

    \item In light-tailed settings (Normal case, case~(b)), the asymptotic variance $v(\tau,h)$ decreases monotonically with~$h$ and converges to its lower bound $\Var(Y)=1$ as $h\to\infty$.
    The absence of a finite minimiser ($h^*(\tau)=\infty$) indicates that the improvement is purely asymptotic:
    smoothing reduces variance gradually but cannot outperform the variance bound attained by the classical quantile estimator.
    Numerically, this appears as efficiency ratios $R(\tau)<1$ that approach~1 for large~$h$.

    \item In heavy-tailed settings (Laplace case, case~(c)), the function $v(\tau,h)$ admits a finite minimiser $h^*(\tau)\in[0.5,0.8]$.
    Moderate smoothing yields substantial efficiency gains for off-central quantiles,
    with $R(\tau)$ decreasing to about~0.45 at $\tau=0.25$ and~$\tau=0.75$.
    At the median, the improvement remains moderate ($R(0.5)\approx0.67$).

    \item The shape of $R(\tau)$ across quantile levels illustrates a clear trade-off between efficiency and robustness:
    heavier tails lead to stronger gains from smoothing, whereas in Gaussian contexts the benefit is marginal and asymptotic.
\end{itemize}

Overall, these results confirm that the smoothing parameter~$h$ acts as a regularisation device:
it stabilises the empirical criterion in the presence of heavy-tailed noise without altering the target quantile,
and its optimal magnitude reflects the tail behaviour of the underlying distribution.

\section{Conclusion}
\label{sec:conclusion}

This paper has introduced and rigorously analyzed a unified family of smoothed quantile estimators that provide a continuous interpolation between the classical empirical quantile and the sample mean. We established the main theoretical properties of this family, including existence and uniqueness of the population minimizer, convexity and interpolation structure, and asymptotic normality with explicit variance characterization.

A key insight is the geometric representation of the parameter space: for each fixed quantile level, admissible parameter pairs $(z,h)$ lie on straight lines along which the population quantile remains constant while the asymptotic efficiency evolves. This geometry clarifies how the smoothing parameter $h$ acts as a regularisation device—reducing asymptotic variance in heavy-tailed or asymmetric settings while preserving the target quantile.

The efficiency analysis revealed two distinct regimes. Under light-tailed distributions (e.g., Gaussian), the variance decreases monotonically in $h$ and asymptotically converges to $\Var(Y)$, implying that the classical quantile estimator is already efficient. Under heavy-tailed distributions (e.g., Laplace), a finite smoothing level $h^*(\tau)>0$ yields a tangible efficiency gain, particularly for off-central quantiles. These theoretical conclusions were numerically validated in Section~\ref{sec:numerical_validation}, confirming the efficiency–robustness trade-off predicted by the model.

The framework developed here extends naturally to quantile regression, where we aim to estimate conditional quantile functions. Future research will investigate smoothed quantile regression by adapting the proposed objective to regression settings, exploring data-driven selection of the smoothing parameter, and establishing the corresponding asymptotic properties of regression coefficients.

Overall, the interpolation framework provides a unified perspective that bridges robust quantile-based estimation with efficient mean-based approaches, offering a flexible and theoretically grounded tool for inference under diverse distributional conditions.

\section*{Declarations}
\paragraph{Data Availability Statement} No real-world data were used in this study. All results are based on simulated data generated by the authors. The simulation code used to produce the results is available from the corresponding author upon reasonable request.

\paragraph{Funding Statement} This research received no specific grant from any funding agency in the public, commercial, or not-for-profit sectors.

\paragraph{Conflict of Interest Statement} The authors declare that there are no conflicts of interest regarding the publication of this paper.
	
\bibliographystyle{plainnat}
\bibliography{references}

\begin{thebibliography}{9}
\providecommand{\natexlab}[1]{#1}
\providecommand{\url}[1]{\texttt{#1}}
\expandafter\ifx\csname urlstyle\endcsname\relax
  \providecommand{\doi}[1]{doi: #1}\else
  \providecommand{\doi}{doi: \begingroup \urlstyle{rm}\Url}\fi

\bibitem[Dermoune et~al.(2017)Dermoune, Ounaissi, and Rahmania]{Dermoune2017}
Azzouz Dermoune, Daoud Ounaissi, and Nadji Rahmania.
\newblock Oscillation of metropolis–hastings and simulated annealing
  algorithms around lasso estimator.
\newblock \emph{Mathematics and Computers in Simulation}, 135:\penalty0 39--50,
  2017.
\newblock ISSN 0378-4754.
\newblock \doi{https://doi.org/10.1016/j.matcom.2015.09.003}.
\newblock URL
  \url{https://www.sciencedirect.com/science/article/pii/S0378475415001901}.
\newblock Special Issue: 9th IMACS Seminar on Monte Carlo Methods.

\bibitem[Embrechts and Hofert(2013)]{Embrechts2013}
Paul Embrechts and Marius Hofert.
\newblock A note on generalized inverses.
\newblock \emph{Mathematical Methods of Operations Research}, 77\penalty0
  (3):\penalty0 423--432, April 2013.
\newblock ISSN 1432-5217.
\newblock \doi{10.1007/s00186-013-0436-7}.

\bibitem[Hampel et~al.(2005)Hampel, Ronchetti, Rousseeuw, and
  Stahel]{Hampel2005}
Frank~R. Hampel, Elvezio~M. Ronchetti, Peter~J. Rousseeuw, and Werner~A.
  Stahel.
\newblock \emph{Robust Statistics: The Approach Based on Influence Functions}.
\newblock Wiley, March 2005.
\newblock ISBN 9781118186435.
\newblock \doi{10.1002/9781118186435}.

\bibitem[Huber(1981)]{Huber1981}
Peter~J. Huber.
\newblock \emph{Robust statistics}.
\newblock Wiley series in probability and mathematical statistics. Wiley, New
  York, 1981.
\newblock ISBN 9780471725244.
\newblock Includes bibliographical references and index.

\bibitem[Koenker(2009)]{Koenker2009}
Roger Koenker.
\newblock \emph{Quantile regression}.
\newblock Number~38 in Econometric Society monographs. Cambridge Univ. Press,
  Cambridge [u.a.], repr. edition, 2009.
\newblock ISBN 9780521608275.
\newblock Literaturverz. S. [319] - 335.

\bibitem[Koenker and Bassett(1978)]{Koenker1978}
Roger Koenker and Gilbert Bassett.
\newblock Regression quantiles.
\newblock \emph{Econometrica}, 46\penalty0 (1):\penalty0 33, January 1978.
\newblock ISSN 0012-9682.
\newblock \doi{10.2307/1913643}.

\bibitem[Tibshirani(1996)]{Tibshirani1996}
Robert Tibshirani.
\newblock Regression shrinkage and selection via the lasso.
\newblock \emph{Journal of the Royal Statistical Society Series B: Statistical
  Methodology}, 58\penalty0 (1):\penalty0 267--288, January 1996.
\newblock ISSN 1467-9868.
\newblock \doi{10.1111/j.2517-6161.1996.tb02080.x}.

\bibitem[Vaart(1998)]{Vaart1998}
A.~W. van~der Vaart.
\newblock \emph{Asymptotic Statistics}.
\newblock Cambridge University Press, October 1998.
\newblock ISBN 9780521784504.
\newblock \doi{10.1017/cbo9780511802256}.

\bibitem[Zou(2006)]{Zou2006}
Hui Zou.
\newblock The adaptive lasso and its oracle properties.
\newblock \emph{Journal of the American Statistical Association}, 101\penalty0
  (476):\penalty0 1418--1429, December 2006.
\newblock ISSN 1537-274X.
\newblock \doi{10.1198/016214506000000735}.

\end{thebibliography}
\end{document}